\documentclass{mathincs}

\usepackage[ansinew]{inputenc}
\usepackage{algorithmic,algorithm}
\usepackage{nomencl}
\usepackage{mathptmx,courier}
\usepackage[scaled]{helvet}
\usepackage{amsmath,amssymb,amsfonts,latexsym}
\usepackage{amsthm}
\usepackage{algorithm,algorithmic}
\usepackage{slashed}
\usepackage{relsize}
\usepackage{caption}
\usepackage{subfigure}

\usepackage[]{color}
\usepackage[colorlinks]{hyperref}
\hypersetup{
 citecolor=blue,
 linkcolor=red,
 urlcolor=magenta
}

 \newtheorem{theorem}{Theorem}[section]
 \newtheorem{lemma}[theorem]{Lemma}
 \newtheorem{corollary}[theorem]{Corollary}
 
 \newtheorem{problem}[theorem]{Problem}
 \theoremstyle{definition}
 \newtheorem{definition}[theorem]{Definition}
 \theoremstyle{remark}

 \numberwithin{equation}{section}

\DeclareMathOperator*{\sq}{\Box}

\newcommand{\eqcl}{\mathrel{\mathsmaller{\mathsmaller{^{\boldsymbol{\sqsubseteq}}}}}}
\newcommand{\la}{\langle}
\newcommand{\ra}{\rangle}

\newcommand{\R}{\mathfrak d}

\begin{document}
\sloppy
%
%
%
%
%
%
%
%
%

\title[Partial Star Products]
 {Partial Star Products: A Local Covering Approach for the Recognition of Approximate Cartesian Product Graphs}
\author{Marc Hellmuth}

\address{%
Center for Bioinformatics \\
Saarland University \\
Building E 2.1, Room 413 \\
P.O. Box 15 11 50 \\
D - 66041 Saarbr\"{u}cken \\
Germany }

\email{marc.hellmuth@bioinf.uni-sb.de}

\author{Wilfried Imrich}
\address{Chair of Applied Mathematics\\
Montanuniversit{\"a}t,  A-8700
Leoben, \\Austria}
\email{imrich@unileoben.ac.at}

\author{Tomas Kupka}
\address{ Chair of Applied Mathematics\\
Montanuniversit{\"a}t,  A-8700
Leoben, \\Austria}
\email{tomas.kupka@teradata.com}

\thanks{This work was financed in part by ARRS Slovenia and the Deutsche
Forschungsgemeinschaft (DFG) within the EUROCORES
Programme EuroGIGA (project GReGAS) of the European Science Foundation. The
Austrian participation in GReGAS is \emph{not} supported by the
\emph{Austrian Science Fonddation (FWF)}.}



\date{\today}

\dedicatory{}

\begin{abstract}
 This paper is concerned with the recognition of
 approximate graph products with respect to the Cartesian
 product. Most graphs are prime, although they can
 have a rich product-like structure. The proposed algorithms
 are based on a local approach that covers a graph by small
 subgraphs, so-called partial star products, and then utilizes
 this information to derive the global factors and an embedding
 of the graph under investigation into Cartesian product graphs.
\end{abstract}

\maketitle

\section{Introduction}

This contribution is concerned with the recognition of approximate products
with respect to the Cartesian product. It is well-known that graphs with
a non-trivial product structure can be recognized in linear
time in the number of edges for Cartesian product graphs \cite{ImPe07}.
Unfortunately, the application of the ``classical'' factorization
algorithms is strictly limited, since almost all graphs are prime, i.e.,
they do not have a non-trivial product structure although they can have a
product-like structure. In fact, even a very small perturbation, such as
the deletion or insertion of a single edge, can destroy the product
structure completely, modifying a product graph to a prime graph
\cite{Fei-86,Zmazek:07}. Hence, an often appearing problem can be
formulated as follows: For a given graph $G$ that has a product-like
structure, the task is to find a graph $H$ that is a non-trivial product
and a good approximation of $G$, in the sense that $H$ can be reached from
$G$ by a small number of additions or deletions of edges and vertices. The
graph $G$ is also called \emph{approximate} product graph.

The recognition of approximate products has been investigated by several
authors, see e.g. \cite{FEHA-89, HIKS-08, HIKS-09, hellmuth2011local,
IZ-96, Zmazek:07,IPZ-97, Zer00, ZZ-02, HZ-99, hos12}. In \cite{IZ-96} and
\cite{Zmazek:07} the authors showed that Cartesian and strong product
graphs can be uniquely reconstructed from each of its one-vertex-deleted
subgraphs. Moreover, in \cite{IZZ-01} it is shown that $k$-vertex-deleted
Cartesian product graphs can be uniquely reconstructed if they have at
least $k+1$ factors and each factor has more than $k$ vertices. In
\cite{IPZ-97, Zer00, ZZ-02} algorithms for the recognition of so-called
graph bundles are provided. Graph bundles generalize the notion of graph
products and can also be considered as a special class of approximate
products. Equivalence relations on the edge set of a graph $G$ that satisfy
restrictive conditions on chordless squares play a crucial role in the
theory of Cartesian graph products and graph bundles. In \cite{hos12} the
authors showed that such relations in a natural way induce equitable
partitions on the vertex set of $G$, which in turn give rise to quotient
graphs that can have a rich product structure even if $G$ itself is prime.
However, Feigenbaum and Haddad proved that the following problem
is NP-complete

\begin{problem}[ \cite{FEHA-89}]
 To a given connected prime graph $G$ find a connected Cartesian product
 $G_1 \Box \dots \Box G_k$ with the same number of vertices as $G$, such
 that $G$ can be obtained from $G_1 \Box \dots \Box G_k$ by adding a
 minimum number of edges only or deleting a minimum number of edges only.
 \label{prob:optimum}
\end{problem}

Hence, in order to solve this problem not only for special classes of
graphs but also for general cases one should provide heuristics that can be
used in order to solve the problem of finding ``optimal'' approximate
products. A systematic investigation into approximate product graphs w.r.t.
the strong product showed that a practically viable approach can be based
on \emph{local} factorization algorithms, that cover a graph by
factorizable small patches and attempt to stepwisely extend regions with
product structures \cite{HIKS-08, HIKS-09, hellmuth2011local}. In the case
of strong product graphs, one benefits from the fact that the local product
structure of induced neighborhoods is a refinement of the global factors
\cite{hellmuth2011local}. However, the problem of finding factorizable
small patches in Cartesian products becomes a bit more complicated, since
induced neighborhoods are not factorizable in general. In order to develop
a heuristic, based on factorizable subgraphs and local coverings which in
turn can be used to factorize large parts of the possibly disturbed graph
we introduce the so-called partial star product (PSP). The partial star
product is, besides trivial cases such as squares, one of the smallest
non-trivial subgraphs that can be isometrically embedded into the product
of so-called stars, even if the respective induced neighborhoods are prime.
Considering a subset of all partial star products of a graph, we
propose in this contribution several algorithms to compute so-called
product colorings and coordinatizations of the subgraph induced by the
partial star products. This information can then be used to embed large
parts of a (possibly) prime graph into a Cartesian product.

We thus present a heuristic algorithm that computes a product 
that \emph{differs as little as possible} from a given graph $G$ and \emph{retains 
as much as possible} of the inherent product structure of $G$. This approach is markedly different 
from the approach of Graham and Winkler \cite{GW:85},
who present a deterministic algorithm that embeds any given, connected graph $G$
isometrically into a Cartesian product $H$. The embedding also
has the remarkable property that any automorphism of $G$ is extends to an automorphisms of $H$. 
Nonetheless, from our point of view, their approach has the disadvantage that $H$ may be exorbitantly large.
For example, if $G$ is a tree on $m$ edges, then the graph $H$ computed by \cite{GW:85} has
$2^m$ vertices.

This contribution is organized as follows. We begin with an introduction
into necessary preliminaries and continue to define the partial star
product. We proceed to give basic properties of the partial star product
and concepts of product relations based on PSP's. These results are then
used to develop algorithms and heuristics that compute (partial)
factorizations of given (un)disturbed graphs.

\section{Preliminaries}

\subsection{Basic Notation}
We consider finite, simple, connected and undirected graphs $G=(V,E)$ with
vertex set $V(G)=V$ and edge set $E(G)=E$. A map $\gamma:V(H)\rightarrow
V(G)$ such that $(x,y) \in E(H)$ implies $(\gamma(x),\gamma(y)) \in E(G)$
for all $x,y \in V(G)$ is a \emph{homomorphism}.
An injective homomorphism $\gamma:V(H)\rightarrow V(G)$
is called  \emph{embedding of $H$ into $G$}.  
We call two graphs $G$ and $H$ \emph{isomorphic}, and write
$G\simeq H$, if there exists a bijective homomorphism $\gamma$ whose
inverse function is also a homomorphism. Such a map $\gamma$ is called an
\emph{isomorphism}.

For two graphs $G$ and $H$ we write $G\cup H$ for the graph $(V(G)\dot\cup
V(H), E(G)\dot\cup E(H))$, where $\dot\cup$ denotes the disjoint union.
The \emph{distance} $d_G(x,y)$ in $G$ is defined
as the number of edges of a shortest path connecting the two vertices
$x,y \in V(G)$.
A graph $H$ is a \emph{subgraph} of a graph $G$, in symbols
$H\subseteq G$, if $V(H)\subseteq V(G)$ and $E(H)\subseteq E(G)$. A
subgraph $H\subseteq G$ is \emph{isometric} if $d_H(x,y)=d_G(x,y)$ for all $x,y\in
V(H)$. For given graphs $G$ and $H$ the embedding $\gamma:V(H)\rightarrow
V(G)$ is an \emph{isometric embedding} if $d_H(u,v) =
d_G(\gamma(u),\gamma(v))$ for all $u,v \in V(G)$. For simplicity, in such case we
also call $H$ isometric subgraph of $G$. If $H\subseteq G$ and all pairs
of adjacent vertices in $G$ are also adjacent in $H$ then $H$ is called an
\emph{induced} subgraph. The subgraph of a graph $G$ that is induced by a
vertex set $W \subseteq V(G)$ is denoted by $\langle W \rangle$. An induced
cycle on four vertices is called \emph{chordless square}. Let the edges
$e=(v,u)$ and $f=(v,w)$ span a chordless square $\la\{v,u,x,w\}\ra$. Then
$f$ is the \emph{opposite} edge of $(x,u)$. The vertex $x$ is called
\emph{top vertex} (w.r.t. the square spanned by $e$ and $f$). A top vertex
$x$ is \emph{unique} if $|N[x] \cap N[v]| = 2$. In other words, a top
vertex $x$ is not unique if there are further squares with top vertex $x$
spanned by the edges $e$ or $f$ together with a third distinct edge $g$.

We define the \emph{open $k$-neighborhood} of a vertex $v$ as the set
$N_k(v) = \{x\in V(G)\mid 0<d_G(v,x) \leq k\}$. The \emph{closed}
$k$-neighborhood is defined as $N_k[v] = N_k(v)\cup \{v\}$. Unless there is a
risk of confusion, an open or closed $k$-neighborhood is just called
$k$-neighborhood and a $1$-neighborhood just neighborhood and we write $N(v)$,
resp. $N[v]$ instead of $N_1(v)$, resp. $N_1[v]$. To avoid ambiguity, we
sometimes write $N_k^G(v)$, resp. $N_k^G[v]$ to indicate that $N_k(v)$,
resp. $N_k[v]$ is taken with respect to $G$.

The \emph{degree} of a vertex $v$ is defined as the cardinality $|N(v)|$.
A \emph{star} $G=(V,E)$ is a connected acyclic graph such that there is a vertex $x$
that has degree $|V|-1$ and the other $|V|-1$ vertices have degree $1$. We
call $x$ the \emph{star-center} of $G$.

\subsection{Product and Approximate Product Graphs}
The Cartesian product $G\Box H$ has vertex set $V(G\Box H)=V(G)\times
V(H)$; two vertices $(g_1,h_1)$, $(g_2,h_2)$ are adjacent in $G\Box H$ if
$(g_1,g_2)\in E(G)$ and $h_1=h_2$, or $(h_1,h_2)\in E(G_2)$ and $g_1 =
g_2$. The one-vertex complete graph $K_1$ serves as a unit, as $K_1 \Box H
\simeq H$ for all graphs $H$. A Cartesian product $G\Box H$ is called
\emph{trivial} if $G \simeq K_1$ or $H \simeq K_1$. A graph $G$ is
\emph{prime} with respect to the Cartesian product if it has only a trivial
Cartesian product representation. A representation of a graph $G$ as a
product $G_1\sq G_2\sq \cdots \sq G_k$ of prime graphs is called a {\it
prime factor decomposition (PFD)} of $G$.

\begin{theorem}[\cite{Sa60, ImPe07}]
Any finite connected graph $G$ has a unique PFD with respect to the
Cartesian product up to the order and isomorphisms of the factors. The PFD
can be computed in linear time in the number of edges of $G$.
\label{thm:upfd}
\end{theorem}

The Cartesian product is commutative and associative. It is well-known that
a vertex $x$ of a Cartesian product $\Box_{i=1}^n G_i$ is properly
``coordinatized'' by the vector $c(x) := (c_1(x),\dots,c_n(x))$ whose
entries are the vertices $c_i(x)$ of its factor graphs $G_i$
\cite{Hammack:2011a}. Two adjacent vertices in a Cartesian product graph
therefore differ in exactly one coordinate. Note, the coordinatization of a
product is equivalent to an edge coloring of $G$ in which edges $(x,y)$
share the same color $c_k$ if $x$ and $y$ differ in the coordinate $k$.
This colors the edges of $G$ (with respect to the \emph{given} product
representation). It follows that for each color $c$ the set $E^c=\{e\in
E(G) \mid c(e)=c\}$ of edges with color $c$ spans $G$. The connected
components of $\langle E^c\rangle$, usually called the \emph{layers} or
\emph{fibers} of $G$, are isomorphic subgraphs of $G$. A \emph{partial
product} $H\subseteq G$ is an isometric subgraph of a (not necessarily
non-trivial) Cartesian product graph $G$.

For later reference, we state the next two well-known lemmas.

\begin{lemma}[Distance Lemma, \cite{IMKL-00}]
  \label{prop:distlemma}
 Let $x=(x_G,x_H)$ and $y=(y_G,y_H)$ be arbitrary vertices of the
  Cartesian product of $G\Box H$. Then
  \[
  d_{G\Box H}(x,y) = d_G(x_G,y_G) + d_H(x_H,y_H)\;.
  \]
\end{lemma}

\begin{lemma}[Square Property, \cite{IMKL-00}]
	Let $G=\Box_{i=1}^n G_i$ be a Cartesian product graph and $e=(u,v),f=(u,w) \in
	E(G)$ be two incident edges that are in different fibers. Then there is
	exactly one square in $G$ containing both $e$ and $f$ and this
	square is chordless. 
\end{lemma}

For more detailed information about product graphs we refer the interested
reader also to \cite{Hammack:2011a,IMKL-00} or \cite{IMKLDO-08}.

For the definition of approximate graph products we defined in
\cite{HIKS-08} the \emph{distance} $d(G,H)$ between two graphs $G$ and $H$
as the smallest integer $k$ such that $G$ and $H$ have representations
$G'$, $H'$, that is vertices in $V(G)$ are identified with vertices
in $V(H)$, for which the sum of the symmetric differences between the
vertex sets of the two graphs and between their edge sets is at most $k$.
That is, if
$$|V(G')\, \triangle\, V(H')|+|E(G')\, \triangle\, E(H')| \leq k.$$

A graph $G$ is a \emph{k-approximate graph product} if there is a non-trivial product
$H$ such that $$d(G,H) \leq k.$$ Here $k$ need not be constant, it can be a
slowly growing function of $|E(G)|$. Moreover, the next results illustrate
the complexity of recognizing approximate graph products.

\begin{lemma}[\cite{HIKS-08}]
\label{lem:k-approx}
For fixed $k$ all Cartesian
$k$-approximate  graph products
can be recognized in polynomial time in $n$.
\end{lemma}

Without the restriction on $k$ the problem of finding a product of closest
distance to a given graph $G$ is NP-complete for the Cartesian product
\cite{FEHA-89}; see Problem \ref{prob:optimum}.

\subsection{Relations}

We will consider equivalence relations $R$ on edge sets $E$, i.e.,
$R\subseteq E\times E$ such that (i) $(e,e)\in R$ (\emph{reflexivity}), (ii)
$(e,f)\in R$ implies $(f,e)\in R$ (\emph{symmetry}) and (iii) $(e,f)\in R$
and $(f,g)\in R$ implies $(e,g)\in R$ (\emph{transitivity}). We will
furthermore write $\varphi \eqcl R$ to indicate that $\varphi$ is an
equivalence class of $R$. A relation $Q$ is \emph{finer} than a relation
$R$  while the relation $R$ is \emph{coarser} than $Q$ if $(e,f)\in Q$
implies $(e,f)\in R$, i.e, $Q\subseteq R$. In case, a given reflexive and
symmetric relation $R$ need not be transitive, we denote with $R^*$ its
transitive closure, that is the finest equivalence relation on $E(G)$ that
contains $R$. For a given graph $G=(V,E)$ and an equivalence relation $R$
on $E$ we define the \emph{$R$-coloring} of $G$ as a map of the edges onto
its equivalence class, i.e, the edge $e\in E$ is assigned color $k$ iff $e\in
\varphi_k\eqcl R$.

For a given equivalence class $\varphi \eqcl R$ and a vertex $u\in V(G)$ we
denote the set of neighbors of $u$ that are incident to $u$ via an edge in
$\varphi$ by $N_{\varphi}(u)$, i.e.,
\begin{equation*}
N_{\varphi}(u):= \{v\in V(G)\mid [u,v]\in\varphi\}\,.
\end{equation*}
The closed $\varphi$-neighborhood is then
$N_{\varphi}[u]=N_{\varphi}(u)\cup\{u\}$.

For later reference we need the following simple lemma.
\begin{lemma}
	Let $R$ be an equivalence relation defined on the edge set of a given
	graph $G=(V,E)$ and $H\subseteq G$ be a subgraph of $G$.
	Then the restriction $R_{|H} = \{(e,f) \in R \mid e,f \in E(H)\}$ of $R$
	on the edge set $E(H)$ is an equivalence relation.
	\label{lem:restrictionEquirel}
\end{lemma}
\begin{proof}
	Clear.
\end{proof}

For the recognition of Cartesian products the relation $\delta$ is of
particular interest.

\begin{definition}
Two edges $e,f\in E(G)$ are in the \emph{relation $\delta(G)$},
if one of the following conditions in $G$ is satisfied:
\begin{itemize}
\item[(i)]  $e$ and $f$ are adjacent and there is no unique square
             spanned by $e$ and $f$ which is in particular chordless. 
\item[(ii)]   $e$ and $f$ are opposite edges of a chordless square.
\item[(iii)] $e=f$.
\end{itemize}
\label{def:delta}
\end{definition}

If there is no risk of confusion we write $\delta$ instead of $\delta(G)$.
Clearly, the relation $\delta$ is reflexive and symmetric but not
necessarily transitive. However, the transitive closure $\delta^*$ is an
equivalence relation on $E(G)$ that contains $\delta$. Note, that our
definition of $\delta$ slightly differs from the usual one, see e.g.
\cite{IZZ-01, Imrich:94}, which is defined analogously without
forcing the chordless square in Condition $(i)$ to be unique.
However, for our
purposes this definition is more convenient and suitable to find the
necessary local information that we use to define those factorizable small
patches which are needed to cover the graphs under investigation and to compute
the PFD or approximations of it with respect to the Cartesian product.
Moreover, as stated in \cite{IZZ-01, Imrich:94}, any pair of adjacent edges that
belong to different $\delta^*$ classes span a unique chordless square,
where $\delta$ is defined without  claiming ``uniqueness'' in Condition $(i)$.
Thus, we can easily conclude that the transitive closure of our relation
$\delta$ and the usual one are identical.

Finally, two edges $e$ and $f$ are in relation $\sigma(G)$ if they have the
same Cartesian colors with respect to the prime factorization of $G$. We
call $\sigma(G)$ the \emph{product relation}. The first polynomial time
algorithm to compute the factorization of a graph explicitly constructs
$\sigma$ starting from the finer relation $\delta$
\cite{FeHeSch85}. The product relation $\sigma$ was later shown to be
simply the convex hull $\mathfrak{C}(\delta)$ of the relation $\delta(G)$
\cite{Imrich:94}.
Notice that $\delta(G) \subseteq \delta(G)^* \subseteq \sigma(G)$ \cite{Imrich:94}.

\section{The Partial Star Product}

\subsection{Basics}

In order to compute $\delta$ from local coverings of the graph $G=(V,E)$ we
need some new notions. Clearly, $\delta$ is still defined in a local manner
since only the (non-)existence of squares are considered and thus, only the
induced $2$-neighborhoods are of central role. However, although the
$2$-neighborhood can be prime, we define subgraphs of $2$-neighborhoods,
that are factorizable or at least graphs that can be isometrically embedded
into Cartesian products and have therefore a rich product structure. For
this purpose we define for a vertex $v\in V(G)$ the relation $\R_v$, that is
a subset of $\delta$ and provides the desired information of the local
product structure of the subgraph $\la N_2[v]\ra$. Based on the transitive
closure $\R^*_v$ we then define the so-called partial star product $S_v$, a
subgraph of $\la N_2[v]\ra$, which provides the details which parts of the
induced 2-neighborhood are factorizable or can be isometrically embedded
into a Cartesian product.

Let $G=(V,E)$ be a given graph, $v \in V$ and $E_v$ be the set of edges
incident to $v$. The local relation $\R_v$ is then defined as $$\R_v = \R(\{v\}) =
((E_v\times E) \cup (E\times E_v)) \cap \delta(G) \subseteq \delta(\la
N_2^G[v]\ra).$$ In other words, $\R_v$ is the subset of $\delta(G)$ that
contains all pairs $(e,f)\in \delta(G)$, where at least one of the edges
$e$ and $f$ is incident to $v$. Note, $\R^*_v$ is not necessarily a
subset of $\delta$ but it is contained in $\delta^*$.

For a subset $W\subseteq V$ we write $\R(W)$ for the union of local relations
$\R_v$, $v\in W$: $$\R(W) = \cup_{v\in W} \R_v.$$

We now define the so-called partial star product $S_v$, that is, a subgraph
containing all edges incident to $v$ and all squares
spanned by edges $e, e'\in E_v$ where $e$ and $e'$ are not in relation $\R_v^*$.
To be more precise:

\begin{definition}[Partial Star Product (PSP)]
Let $F_v\subseteq E\setminus E_v$
be the set of edges which are opposite edges of (chordless)
squares spanned by $e,e'\in E_v$ that are in different
$\R^*_v$ classes, i.e., $(e,e') \not\in \R^*_v$.

The \emph{partial star product} is the subgraph
$S_v \subseteq G$ with edge set $E'= E_v\cup F_v$ and vertex set $\cup_{e\in E'}
e$. We call $v$ the \emph{center} of $S_v$, edges in $E_v$ \emph{primal
edges}, edges in $F_v$ \emph{non-primal edges}, and the vertices adjacent
to $v$ \emph{primal vertices} with respect to $S_v$.
	\label{def:starproduct}
\end{definition}

\begin{figure}[tbp]
  \centering
  \includegraphics[bb= 131 286 477 501, scale=0.7]{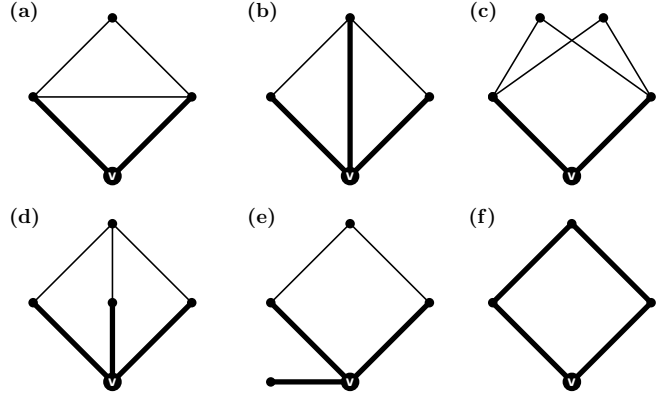}
 \caption{Examples of various PSP's $S_v$ highlighted by thick edges.
		   Note, in all cases except in
		   case $(f)$ the set $F_v$ is empty and hence, the PSP's $S_v$
			in the other cases just contain the edges incident to $v$.
           }
	\label{fig:PSP_definition}
\end{figure}

The reason why we call $S_v$ a partial star product is that $S_v$ is an
isometric subgraph or even isomorphic to a Cartesian product graph $H$ of
stars, as we shall see later (Theorem \ref{thm:isomSubgraph}). Hence, $S_v$
is a partial product of $H$. For the construction of this graph $H$ we introduce
the so-called star factors $\mathbb S_i$, see also Figures
\ref{fig:PSP_definition} and \ref{fig:StarVSn2}.

\begin{figure}[tbp]
  \centering
  \includegraphics[bb= 65 443 503 626, scale=0.5]{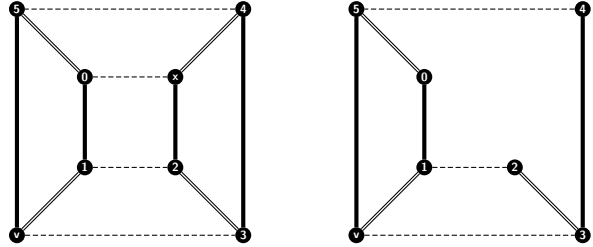}
 \caption{\textbf{Left:} A hypercube $Q_3$ is shown. The three equivalence classes
						of $\delta^*(Q_3)$ are highlighted by solid, dashed and
					  double lined edges, respectively.
						\textbf{Right:} The PSP $S_v$
						is shown. Again, $\R^*_{v|S_v}$ has three equivalence classes.
						However, since the edges $(0,1)$ and $(1,2)$ as well as the
						edges $(2,3)$ and $(3,4)$ span no square we can conclude that
						$\delta^*(S_v)$ just contains one equivalence class. Hence,
						 $\R^*_{v|S_v} \neq \delta^*(S_v)$	.}
	\label{fig:Q3}
\end{figure}

\begin{figure}[tbp]
  \centering
  \includegraphics[bb= 124 422 409 647, scale=0.7]{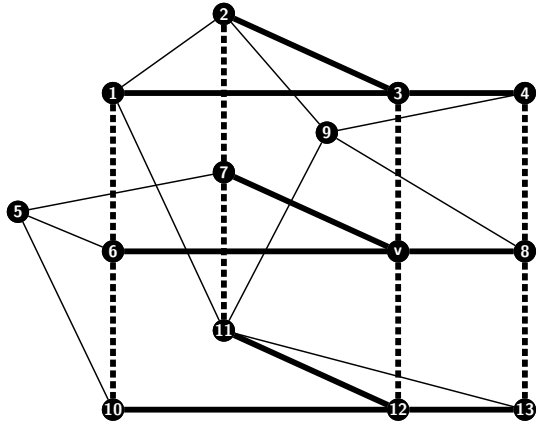}
 \caption{Shown is a graph $G\simeq \la N^G_2[v]\ra$. Note, $\delta(G)^*$
          has one equivalence class and thus, $G$ is prime. However, the
          partial star product (PSP) $S_v$, that is the subgraph that
          consists of thick and dashed edges is not prime. The subgraph
          $S_v$ is isomorphic to the Cartesian Product of a star with four
          and a star with three vertices. The two equivalence classes of
          $\R_{|S_v}$ are highlighted by thick, resp. dashed edges.}
	\label{fig:StarVSn2}
\end{figure}

\begin{definition}[Star Factor]
Let $G=(V,E)$ be an arbitrary given graph and $S_v$ be a PSP for some
vertex $v\in V$. Assume $\R_v^*$ has equivalence classes $\varphi_1, \dots,
\varphi_n$. We define the star factor $\mathbb S_i$ as the graph with
vertex set $N_{\varphi_i}[v]$ that contains all primal edges of $E_v$ that
are also in the induced closed $\varphi_i$-neighborhood, i.e., $E(\mathbb
S_i)= E(\la N_{\varphi_i}[v] \ra) \cap E_v$.
\label{def:starfactor}
\end{definition}

Note, this definition forbids triangles in $\mathbb S_i$, and hence, each
$\mathbb S_i$ is indeed a star. We denote the restriction of $\R^*_v$ to
the subgraph $S_v$ with $$\R_{|S_v}:=\R^*_{v|S_v}= \{(e,f) \in \R^*_v \mid
e,f \in E(S_v)\}.$$ In other words, $\R_{|S_v}$ is the subset of $\R^*_v$
that contains all pairs of edges $(e,f)\in \R^*_v$ where both edges $e$ and
$f$ are contained in $S_v$. We want to emphasize that $\R^*_{v|S_v} \neq
\delta^*(S_v)$; see Figure \ref{fig:Q3}. In addition, by Lemma
\ref{lem:restrictionEquirel} we can conclude that $\R_{|S_v}$ is an
equivalence relation. For a given subset $W\subseteq V$ we define
$$\R_{|S_v}(W) = \cup_{v\in W} \R_{|S_v}$$ as the union of relations
$\R_{|S_v}$, $v\in W$. As it will turn out, for a given graph $G=(V,E)$ the
transitive closure $\R_{|S_v}(V)^*$ is the equivalence relation
$\delta(G)^*$, see Theorem \ref{thm:union_equals_delta}.

\subsection{Properties of the Partial Star Product}

We now establish basic properties of the graph $S_v$, its edge sets $E_v$
and $F_v$, as well as of the relation $\R_v^*$ and its restriction $\R_{|S_v}$
to $S_v$.

\begin{lemma}
  Given a graph $G=(V,E)$ and a vertex $v\in V$. Then
	$F_v = \emptyset$ if and only if for all edges $e,e'\in E_v$
	holds $(e,e')\in \R^*_v$. Moreover, if $F_v \neq \emptyset$
	then $|F_v|\geq 2$.
   \label{lem:psp0}
\end{lemma}

\begin{proof}
	Clearly, if for all edges $e,e'\in E_v$ holds $(e,e')\in \R^*_v$ then by
	definition $F_v=\emptyset$.

	Let $F_v = \emptyset$ and assume there are edges $e,e'\in E_v$ that are
	not in relation $\R_v^*$. In particular, these edges are not in relation
	$\R_v$, and therefore not in relation $\delta(G)$. By Condition $(i)$ of
	Def. \ref{def:delta} and since $e$ and $e'$ are adjacent, there is a
	chordless square containing $e$ and $e'$ and therefore, respective
	opposite edges $f$ and $f'$. Condition $(ii)$ of Def. \ref{def:delta}
	implies $(e,f), (e',f') \in \delta(G)$. Therefore, $f,f' \in F_v$, a
	contradiction.

	Furthermore, since $F_v$ contains all opposite edges of squares spanned by
	$e, e' \in E_v$ we can easily conclude that $|F_v|\geq 2$, if $F_v \neq
	\emptyset$.
\end{proof}

\begin{lemma}
\label{lem:PSP1}
Let G=(V,E) be a given graph and let $S_v$ be a PSP for some vertex $v\in V$.
If  $e, f \in E_v$ are primal edges that are not in relation $\R_v^*$, then
$e$ and $f$ span a unique chordless square with a unique top vertex in $G$.

Conversely, suppose that $x$ is a non-primal vertex of $S_v$,
then there is a unique chordless square in $S_v$ that contains vertex $x$
and that is spanned by edges $e, f \in E_v$ with $(e,f)\not\in \R^*_v$.
\end{lemma}
\begin{proof}
	First, we show that $e$ and $f$ span a unique chordless square in $G$. By
	contraposition, assume $e$ and $f$ span no unique chordless square in $G$. Since
	$e$ and $f$ are adjacent, Condition $(i)$ of Def. \ref{def:delta} implies
	that $(e,f) \in \delta(G)$ and hence, $(e,f) \in \R_v \subseteq \R_v^*$.
	Therefore, if $(e,f) \notin \R_v^*$, then they must span a unique chordless
	square. Let $e=(v,u)$ and $f=(v,w)$, $(e,f) \notin \R_v^*$, span the
	unique chordless square $SQ_1 =\langle\{ v,u,x,w\} \rangle$ and assume
	for contradiction that the top vertex $x$ is not unique. Hence, there must be
	at least three squares: the square $SQ_1$, the square $SQ_2=\langle\{
	v,u,x,y\} \rangle$ spanned by $e$ and $g$, and the square $SQ_3=\langle\{
	v,w,x,y\} \rangle$ spanned by $f$ and $g=(v,y)$. We denote edges as
	follows: $a=(x,y)$ and $b=(x,w)$. Assume both squares $SQ_2$ and $SQ_3$
	are chordless. Then Def. \ref{def:delta} $(ii)$ implies $(f,a), (a,e)\in
	\delta(G)$ and therefore, $(e,f) \in \R_v^*$, a contradiction. If both
	squares have a chord then Def. \ref{def:delta} $(i)$ implies that $(e,g),
	(f,g)\in \delta(G)$ and thus, $(e,f)\in \R_v^*$, again a contradiction. If
	only one square, say $SQ_2$, has a chord $(u,y)$, then $(e,g)\in
	\delta(G)$ and $(f,a), (g,a)\in \delta(G)$ and again we have $(e,f) \in
	\R_v^*$.

	Assume $x$ is a non-primal vertex in $S_v$. By definition, there are
	non-primal edges $f'=(x,u), e'=(x,w) \in F_v$ that are contained in a
	square spanned by $e= (v,u), f=(v,w) \in E_v$, whereas $(e,f)\not\in
	\R_v^*$. As shown above, the square spanned by $e$ and $f$ is unique with
	unique top vertex in $G$ and therefore in $S_v$.
	Hence, if there is another square in $S_v$ containing
	$x$ then it must be spanned by $e',f'$ and this square contains
	additional edges $f''=(y,u), e''=(y,w)$. However, then there is a square
	$\langle\{ v,u,y,w\} \rangle$, which contradicts the fact that the square
	spanned by $e$ and $f$ is unique. If the unique square spanned by $e$ and
	$f$ is not chordless in $G$, then Def. \ref{def:delta} $(i)$ implies $(e,f)\in
	\delta(G)$ and thus $(e,f)\in \R_v^*$, a contradiction.	
\end{proof}

By means of Lemma \ref{lem:psp0} and \ref{lem:PSP1} and the definition of
partial star products we can directly infer the next corollary.

\begin{corollary}
Let G=(V,E) be a given graph and let $S_v$ be a PSP for some vertex $v\in V$.
\begin{enumerate}
\item If  $(e,f)\in \R_v^*$ then there is no square in $S_v$ spanned by $e$ and $f$.
\item Every square in $S_v$ contains two edges $e,e' \in E_v$ and two edges $f,f' \in F_v$,
       and every edge $f \in F_v$ is opposite to some primal edge $e \in E_v$.
\item Every non-primal vertex in $S_v$ is a unique top vertex of some square spanned
       by edges $e, e' \in E_v$.
\end{enumerate}
\label{cor:PSP1}
\end{corollary}

\begin{lemma}
Let G=(V,E) be a given graph and let $f \in F_v$ be a non-primal edge of a
PSP $S_v$ for some vertex $v\in V$. Then $f$ is opposite to exactly one
primal edge $e\in E_v$ in $S_v$ and $(e,f)\in \R_{|S_v}$.
\label{lem:PSP2}
\end{lemma}
\begin{proof}
By Corollary \ref{cor:PSP1}, construction of $S_v$ and since $f\in F_v$,
there is at least one edge $e\in E_v$ such that $f$ is opposite to $e$
and therefore at least one square $SQ_1=\la \{v,w,x,u\} \ra$ in $S_v$
spanned by primal edges $e=(v,u)$ and $e'=(v,w)$ that contains the
edge $f=(w,x)$. Note, by construction $(e,e')\not\in \R^*_v$ and
$e$ is opposite to $f$.
Assume for contradiction that $f$ is opposite to another edge $g=(v,y)$.
Then there is another square $SQ_2=\la\{v,y,x,w\}\ra$.
Hence, $e$ and $e'$ do not span a square with unique top vertex in $G$.
By Definition \ref{def:delta} and Lemma \ref{lem:PSP1} we can conclude that
 $(e,e')\in \R^*_v$, a contradiction.
Hence $e$ and $e'$ span a unique chordless square containing the
edge $f$. By Condition (i) of Definition \ref{def:delta} it holds $(e,f)\in\delta$.
Since $e\in E_v$ we claim $(e,f)\in \R_v$ and consequently $(e,f)\in \R_{|_{S_v}}$.
\end{proof}

\begin{lemma}
Let G=(V,E) be a given graph with maximum degree $\Delta$
and $W\subseteq V$ such that $\la W \ra$ is connected.
Then each vertex $x\in W$ meets every equivalence class
of $\R_{|S_v}(W)^*$ in $\cup_{v\in W} S_v$, i.e.,
for each equivalence class $\varphi\eqcl \R_{|S_v}(W)^*$ and
for each vertex $x\in W$ there is an edge $(x,y)\in \varphi$
with $(x,y)\in E(\cup_{v\in W} S_v)$.
Moreover, $\R_{|S_v}(W)^*$ has at most $\Delta$ equivalence classes.
\label{lem:vMeetsEveryClass}
\end{lemma}
\begin{proof}
 Let $v\in W$ be an arbitrary vertex and $S_v$ be its PSP. We show first
 that $v$ meets every equivalence class of $\R_{|S_v}$ in $S_v$. Assume for
 contradiction that there is an equivalence class $\varphi \eqcl \R_{|S_v}$
 that is not met by $v$ and hence for all edges $e\in E_v$ we have $e\not \in
 \varphi$. Hence, there must be a non-primal $f\in F_v$ with $f \in
 \varphi$. By construction of $S_v$ and by Lemma \ref{lem:PSP2} this edge
 $f$ is opposite to exactly one edge $e\in E_v$ with $(e,f)\in \R_{|S_v}$,
 but then $e\in \varphi$, a contradiction. We show now that every primal
 vertex $w$ in $S_v$ meets every equivalence class of $\R_{|S_v}$. Let
 $\varphi \eqcl \R_{|S_v}$ be an arbitrary equivalence class. If $e =
 (v,w)\in \varphi$ we are done. Therefore assume $e\not \in \varphi$.
 Hence, there must be at least a second equivalence class $\varphi' \eqcl
 \R_{|S_v}$ with $e\in \varphi'$. Since vertex $v$ meets every equivalence
 class there is an edge $e' = (v,u)\in \varphi$. Moreover, since
 $(e,e')\not\in \R_v^*$ it follows that $(e,e')\not\in \R_v \subseteq
 \delta$. Since $e$ and $e'$ are adjacent and by Condition $(i)$ of
 Definition \ref{def:delta} the edges $e$ and $e'$ span a unique chordless
 square. Hence, there is an opposite edge $f=(w,x)$ of $e'$. By
 construction of $S_v$ we have $f\in F_v$ and hence, Lemma \ref{lem:PSP2}
 implies $(e',f)\in \R_{|S_v}$. Therefore, the primal vertex $w$ meets
 equivalence class $\varphi$ in $S_v$. Note, not every equivalence class of
 $\R_{|S_v}$ must be met by non-primal vertices in $S_v$ in general, as
 one can easily verify by the example in Figure \ref{fig:StarVSn3}.

 It remains to show that every vertex $x\in W$ meets every equivalence
 class of $\R_{|S_v}(W)^*$ in $\cup_{v\in W} S_v$. Assume we have chosen an
 arbitrary vertex $x\in W$, computed $S_x$ and $\R_{|S_x}$. As shown, vertex
 $x$ and all its primal neighbors $y$ in $S_x$ meet every equivalence class
 of $\R_{|S_x}$. Assume $W$ contains more than one vertex. Since $\la W\ra$
 is connected there is a primal vertex $y$ of $x$ that is also contained in
 $W$. Hence, vertex $x$ is a primal neighbor of $y$ in $S_y$ and every
 equivalence class of $\R_{|S_y}$ is met by $y$ as well as by $x$. Let
 $\varphi \eqcl (\R_{|S_x} \cup \R_{|S_y})^*$ be an arbitrary equivalence
 class. Assume neither $x$ nor $y$ meets $\varphi$. Then each edge $f\in
 \varphi$ must be in $F_x$ or $F_y$. Assume $f\in F_y$ then, by construction
 of $S_y$ and Lemma \ref{lem:PSP2}, this edge $f$ is opposite to exactly one
 edge $e\in E_y$ with $(e,f)\in \R_{|S_y}$, and hence $e\in \varphi$, a
 contradiction. Assume now all edges $e\in \varphi$ are only met by $y$ but
 not by $x$, and therefore, $e'=(x,y)\not \in \varphi$. However, since $e$
 and $e'$ are in different equivalence classes of $(\R_{|S_x} \cup
 \R_{|S_y})^*$ they must be in different equivalence classes of $\R_{|S_y}$.
 Hence, $(e,e') \not \in \R^*_y$ and thus, $(e,e') \not \in \R_y \subseteq
 \delta$. Since $e$ and $e'$ are adjacent and, by Condition $(i)$ of
 Definition \ref{def:delta}, the edges $e$ and $e'$ span a unique chordless
 square. Hence, there is an opposite edge $f=(x,w)$ of $e$ in $S_y$ and, by
 Lemma \ref{lem:PSP2} we conclude $(e,f)\in \R_{|S_y}$ and therefore, $f\in
 \varphi$, which implies that $x$ meets $\varphi$, a contradiction. Hence,
 every equivalence class $\varphi \eqcl (\R_{|S_x} \cup \R_{|S_y})^*$ must be
 met by $x$ and $y$. By the same arguments one shows that each primal
 vertex of $S_x$ and $S_y$ meets every equivalence class of
 $(\R_{|S_x} \cup \R_{|S_y})^*$. If $W\setminus \{x,y\} \neq \emptyset$ we
 can choose a primal neighbor $z\in W$ of $x$ or $y$, since $\la W\ra$ is
 connected. By the same arguments as before, one shows that each vertex
 $x,y$, resp. $z$ and each of its primal vertices in $S_x, S_y$, resp.
 $S_z$ meets every equivalence class of $((\R_{|S_x} \cup \R_{|S_y})^* \cup
 \R_{|S_z})^* = (\R_{|S_x} \cup \R_{|S_y} \cup \R_{|S_z})^*$ in $S_x \cup S_y
 \cup S_z$. Therefore, we can traverse $\la W\ra$ in breadth-first search order
 and inductively conclude that every vertex $x\in W$ meets every
 equivalence class of $\R_{|S_v}(W)^*$ in $\cup_{v\in W} S_v$.

 Finally, we observe that each edge in $E_v$ might define one equivalence
 class of $\R_{|S_v}$ for each vertex $v\in W$. Thus, $\R_{|S_v}$ can have at
 most $\Delta$ equivalence classes. Since this holds for all vertices and
 since equivalence classes in $\R_{|S_v}(W)^*$ are combined equivalence
 classes of the respective $\R_{|S_v}$ classes, the number of equivalence
 classes in $\R_{|S_v}(W)^*$ can not exceed $\Delta$.
\end{proof}

In order to prove that each PSP can be isometrically embedded into a
Cartesian product of stars, which is shown in the next theorem, we first
need the following lemma.

\begin{lemma}
	Let  $G = \Box_{i=1}^l G_i$ be the Cartesian product
	of stars. Assume the vertices in each $V(G_i)$ are labeled
	from $0,\dots,|V(G_i)|-1$, where the vertex with label
	$0$ always denotes the star-center of each $G_i$.
	Let $v_G$ be the vertex with coordinates $c(v_G)= (0,\dots,0)$
	Then for any integer $k\geq 0$, the induced closed $k$-neighborhood $\la N_k^G[v_G] \ra$
	is an isometric subgraph of $G$.
	\label{lem:isomSubgraph}
\end{lemma}
\begin{proof}
	Let $\la N_k^G[v_G] \ra$ be the induced closed $k$-neighborhood of
	$v_G$ in $G$. Let $x, y \in N_k^G[v_G]$ be arbitrary vertices and let
	$I\subseteq \{1,\dots,l\}$ be the set of positions where $x$ and $y$
	differ in their coordinate. Moreover, let $I_0\subseteq I$ be the set of
	positions where either $x$ or $y$ has coordinate $0$. By the Distance
	Lemma we have
  $d_G(x,y)=\sum_{i\in I_O} 1 + \sum_{i\in I \setminus I_O} 2$.
		
	We now construct a path from $x$ to $y$ that is entirely contained in $N_k^G[v_G]$
	and show that this path is a
	shortest path. Set  $P(x,y)=\emptyset$.
  Let $i\in I_0$ and w.l.o.g. assume $c_i(x) = 0$, otherwise
	we would interchange the role of $x$ and $y$. By definition of the Cartesian
  product there is a vertex $y'$ that is adjacent to
	vertex $y$ with $c_j(y')=c_j(y)$ for all $j\neq i$ and $c_i(y')=0$. By
	the Distance Lemma, we have $d_{G_j}(c_j(v_G), c_j(y)) =
	d_{G_j}(c_j(v_G), c_j(y'))$ for all $j\neq i$ and $d_{G_i}(c_i(v_G),
	c_i(y))=d_{G_i}(0, c_i(y))=1$ and $d_{G_i}(c_i(v_G), c_i(y'))=0$ and
	thus, $d_G(v_G,y') < d_G(v_G,y) \leq k$, which implies that $y'\in
	N_k^G[v_G]$. We assign $(y,y')$ to be an edge of the (so far empty)
	path $P(x,y)$ from $x$ to $y$ 
  and repeat to construct parts of the path from $x$ to $y'$ in the same way
	until all $i\in I_0$ are processed. In this way, we constructed subpaths
	$P(x,v)$ and $P(w,y)$ of $P(x,y)$, both of which are entirely contained in $\la
	N_k^G[v_G] \ra$ and $|P(x,v)| + |P(w,y)|=|I_0|$. We are left to construct
	a path from $v$ to $w$ that is entirely contained in $N_k^G[v]$. Note
	that by construction $v$ and $w$ differ only in the $i$-th position of
	their coordinates where $i \in I\setminus I_0$ and $c_j(v)=c_j(x)=c_j(y)=c_j(w)$
	for all $j\not \in I\setminus I_0$. By the definition of the Cartesian product
	for each $i \in I\setminus I_0$ there are edges
	$(v,v')$, resp. $(v',v'')$ such that $v, v'$ and $v''$ differ
	only in the $i$-th position of their coordinates. Since $0\neq c_i(x) =
	c_i(v)$ and by definition of the Cartesian product it follows that
	$c_i(v')=0$ and $v''$ can be chosen such that $c_i(v'')=c_i(y)=c_i(w)\neq
	0$. By the Distance Lemma and the same arguments as used before it holds
	$d_G(v_G,v')= d_G(v_G,v'') -1 = d_G(v_G,v)-1\leq k$ and hence, $v',v''\in
	N_k^G[v_G]$. Therefore we add the edges $(v,v')$, resp. $(v',v'')$ to
	the path from $x$ to $y$, remove $i$ from $I\setminus I_0$ and repeat
	this construction for a path from $v''$ to $w$ until $I\setminus I_0$ is
	empty.

	Hence we constructed a path of length $|I_0| + 2|I\setminus I_0| =
	\sum_{i\in I_O} 1 + \sum_{i\in I \setminus I_O} 2 =d_G(x,y)$. Thus, this
	path is a shortest path from $x$ to $y$. Since this construction can be
	done for any $x,y \in N_k^G[v_G]$ we can conclude that $\la N_k^G[v_G]
	\ra $ is an isometric subgraph of $G$.
\end{proof}

\begin{theorem}
	Let $G=(V,E)$ be an arbitrary given graph and $S_v$ be a PSP for some
	vertex $v\in V$. Let $H = \Box_{i=1}^k \mathbb S_i$ be the Cartesian
	product of the star factors as in Definition \ref{def:starfactor}. Then
	it holds:
	\begin{enumerate}
	\item[(1)] $S_v$ is an isometric subgraph of $H$ and
				in particular, $S_v\simeq \la N^H_2[(v_1,\dots, v_k)]\ra$
				where $v_i$ denotes the star-center of $\mathbb S_i$, $i=1,\dots, k$.
	\item[(2)] $\R_{|S_v}\subseteq \delta(H)^* \subseteq \sigma(H)$.
	\item[(3)] The product relation $\sigma(H)$ has the same number of equivalence classes as $\R_{|S_v}$.
	\end{enumerate}
	\label{thm:isomSubgraph}
\end{theorem}
\begin{proof}
	\noindent\textbf{Assertion (1):}\\
	If $\R_v^*$ has only one equivalence class, then there is nothing to show,
	since $S_v \simeq \mathbb S_1 \simeq H$. Therefore, assume $\R_v^*$ has
	$k\geq 2$ equivalence classes. 	

	In the following we define a mapping $\gamma:V(S_v) \rightarrow V(H)$ and
	show that $\gamma$ is an isometric embedding. In particular we show that
	$\gamma$ is an isomorphism from $S_v$ to the $2$-neighborhood $\la
	N_2^H[v_H] \ra$ for a distinguished vertex $v_H\in V(H)$. Lemma
	\ref{lem:isomSubgraph} implies then that this embedding is isometric.

\begin{figure}[tbp]
  \centering
  \includegraphics[bb= 82 422 452 647, scale=0.6]{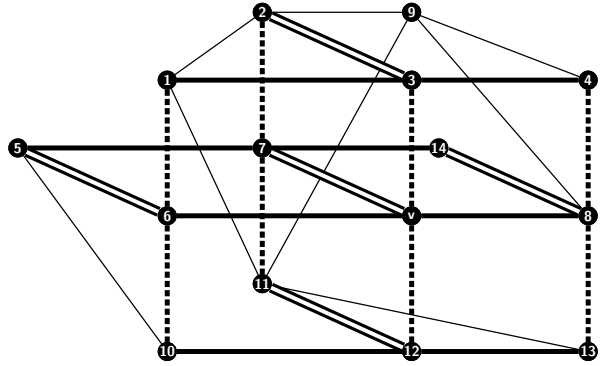}
 \caption{Shown is a graph $G\simeq \la N^G_2[v]\ra$. Note, $\delta(G)^*$ has
						one equivalence class. The partial star product (PSP) $S_v$ is the
						subgraph that consists of thick, double-lined and dashed edges. Moreover, $S_v$
						can be isometrically embedded into the Cartesian product
						of a star with two and two stars with
						three vertices. The three equivalence classes of
						$\R_{|S_v}$ are highlighted by thick, double-lined, resp. dashed edges.}
	\label{fig:StarVSn3}
\end{figure}

	For a given equivalence class $\varphi_i \eqcl \R_v^*$ let
	$N_{\varphi_i}(v) =\{v_1,\dots,v_l\}$ be the $\varphi_i$-neighborhood of
	the center $v$ and $\mathbb S_i$ be the corresponding star factor with
	vertex set $V(\mathbb S_i) = \{0,1,\dots,l\}$ and edges $(0,x)\in
	E(\mathbb S_i)$ for all $(v, v_x)\in S_v$. Let $H = \Box_{i=1}^k \mathbb
	S_i$ be the Cartesian product of the star factors. The center $v$ of
	$S_v$ is mapped to the vertex $v_H\in V(H)$ with coordinates $c(v_H) =
	(0,\dots,0)$, the vertices $v_j \in N_{\varphi_i}(v)$ are mapped to the
	unique vertex $u$ with coordinates $c_r(u) = 0$ for all $r\neq i$
	and $c_i(u) = j$. Clearly, these vertices exist, due to the construction
	of $\mathbb S_1, \dots, \mathbb S_k$ and since $V(H) = \times_{i=1}^k
	V(\mathbb S_i)$. Note, that these vertices we mapped onto are entirely
	contained in the 1-neighborhood $N^H[v_H]$ of $v_H$.
	Now let $x$ be a non-primal vertex in $S_v$. Hence, by Lemma
	\ref{lem:PSP1} and Corollary \ref{cor:PSP1}, there is a unique chordless
	square $\la\{v, v_i, x, v_j\} \ra$ in $S_v$ with unique top vertex $x$.
	Thus, $v_i$ and $v_j$ are the only common neighbors of $x$ in $S_v$.
	Moreover, by definition and Lemma \ref{lem:PSP1}, the edges $(v,v_i)\in
	\varphi_r$ and $(v,v_j)\in \varphi_s$ are in different equivalence classes, i.e.,
	$r \neq s$. Thus, we map $x$ to the unique vertex $u$ with coordinates
	$c_l(u) = 0$ for all $l\neq r,s$ and $c_r(u) = i$ and $c_s(u)= j$. Again,
	this vertex exists, due to the construction of $\mathbb S_1, \dots,
	\mathbb S_k$ and since $V(H) = \times_{i=1}^k V(\mathbb S_i)$. This
	completes the construction of our mapping $\gamma$.

	We continue to show that the mapping $\gamma: V(S_v)\rightarrow
	N_2^H[v_H]$ is bijective. It is easy to see that by construction and
	the definition of the Cartesian product, each primal vertex $x$ has a unique
	partner $\gamma(x)$ in $N_1^H[v_H]$ and vice versa. We show that this
	holds also for non-primal vertices in $S_v$ and vertices in $N_2^H[v_H]
	\setminus N_1^H[v_H]$. First assume there are two non-primal vertices $x$
	and $x'$ in $S_v$ that are mapped to the same vertex $u$ in $H$. Thus, by
	construction of our mapping $\gamma$, the vertex $x'$ must have the same
	primal neighbors $v_i$ and $v_j$ as $x$ in $S_v$. However, by Lemma
	\ref{lem:PSP1} this contradicts that $(v,v_i)\in \varphi_r$ and
	$(v,v_j)\in \varphi_s$ span a unique square. Therefore, $\gamma$ is
	injective. Now, let $u\in N_2^H[v_H] \setminus N_1^H[v_H]$ be an
	arbitrary vertex in $H$. By the Distance Lemma we can conclude that
	$d_H(v_H,u) = \sum_{i=1}^k d_{\mathbb S_i}(0,c_i(u))$. Moreover, since
	$d_H(v_H,u)=2$ and $d_{\mathbb S_i}(0,c_i(u))\leq 1$ for all
	$i=1,\dots,k$ we can conclude that $d_H(v_H,u) =d_{\mathbb S_r}(0,c_r(u))
	+ d_{\mathbb S_s}(0,c_s(u))$ for some distinct indices $r$ and $s$.
	Assume that $c_r(u)=i$ and $c_s(u)=j$. By construction, the star factor
	$\mathbb S_r$ contains the edge $(0,i)$ and $\mathbb S_s$ the edge
	$(0,j)$. Hence, there are edges $e= (v,v_i) \in \varphi_r$ and $f=(v,v_j)
	\in \varphi_s$ in $S_v$. Lemma \ref{lem:PSP1} implies that there is a
	unique chordless square spanned by $e$ and $f$ with unique top vertex $y$
	that is also contained in $S_v$. By construction of $\gamma$ the vertex
	$y$ is the unique vertex that is mapped to vertex $u$ in $H$. Since this
	holds for all vertices $u\in N_2^H[v_H] \setminus N_1^H[v_H]$, and by the
	preceding arguments, we can conclude that the mapping $\gamma:S_v
	\rightarrow N_2^H[v_H]$ we defined is bijective.

	It remains to show that $\gamma$ is an isomorphism from $S_v$ to
	$N_2^H[v_H]$. By construction, every primal edge $(v,v_j) \in \varphi_r$
	is mapped to the edge $(v_H,x)$, where $x$ has coordinates $c_i(x)=0$ for
	$i\neq r$ and $c_r(x) = j$. Hence, $(v,v_j)\in E_v$ if and only
	if $(\gamma(v),\gamma(v_j))\in E(\la N_2^H[v_H]\ra)$. Now suppose we have
	a non-primal edge $(v_j,y) \in \varphi_r$. By Lemma \ref{lem:PSP1}, there
	is a unique chordless square with edges $(v,v_l)\in \varphi_r$ and
	$(v,v_j)\in \varphi_s$ and hence, by construction of $\mathbb S_r$ and
	$\mathbb S_s$ and the definition of the Cartesian product, there are edges
	$e=(v_H,z)$ and $f=(v_H,z')$ in $H$ where $z$ differs from $v_H$ in the
	$r$-th position of its coordinate and $z'$ differs from $v_H$ in the $s$-th
	position of its coordinate. By the Square Property, there is unique
	chordless square in $H$ spanned by $e$ and $f$ with top vertex $y'$ that
	has coordinates $c_i(y')=0$ for $i\neq r,s$, $c_r(y')=l\neq 0$ and
	$c_s(y')=j\neq 0$. By the construction of $\gamma$ we see that $(v_j,y)\in
	F_v$ implies $(\gamma(v_j),\gamma(y))=(z',y')\in E(\la N_2^H[v_H]\ra)$. Using
	the same arguments, but starting from squares spanned by $e=(v_H,z)$ and
	$f=(v_H,z')$ in $H$, one can easily derive that $(z',y')\in E(\la
	N_2^H[v_H]\ra)$ implies $(\gamma^{-1}(z'),\gamma^-1(y'))=(v_j,y)\in F_v$.

	Finally, Lemma \ref{lem:isomSubgraph} implies that $\la N_2^H[v_H]\ra$ is
	an isometric subgraph of $H$ and therefore, $\gamma:V(S_v) \rightarrow
	V(H)$ is an isometric embedding.

\noindent\textbf{Assertion (2) and (3):}\\ 
	By Assertion (1), 
	we can treat the graph $S_v$ as subgraph of $H$; $S_v\subseteq H$.
	We continue to show that $\R_{|S_v}=\R^*_{v|S_v}\subseteq\delta(H)^*$.
	Let $v\in V(G)$ be the
	center of the PSP $S_v$, and $H = \Box_{i=1}^k \mathbb S_i$, where
	$\mathbb S_i$ are the corresponding star factors (w.r.t. $S_v$). Let
	$e,f\in E(S_v)$ such that $(e,f)\in \R_{|S_v}$.
	There are three cases to consider; either
	$e,f\in E_v$, or $e,f\in F_v$, or $e\in E_v$ and $f\in F_v$.

	If $e,f	\in E_v$ are both primal edges with $(e,f)\in \R_{|S_v}$
  then $e$ and $f$ are by construction of the star factors and $H$
  contained in the layer $\mathbb S_i^v$ of some star factor $\mathbb S_i$.
	Corollary \ref{cor:PSP1} and $(e,f)\in \R_{|S_v}\subseteq \R_v^*$
  imply that $e$ and $f$ span no square in $S_v$.
	Since $H = \Box_{i=1}^k \mathbb S_i$ we can conclude that
	$e$ and $f$ span no square in $H$ and hence, $(e,f)\in \delta(H)$.

  Assume $e,f	\in F_v$ and  $(e,f)\in \R_{|S_v}$.
  By Lemma \ref{lem:PSP2} it holds that $e$, resp., $f$ is opposite to exactly one
	primal edge $e'\in E_v$, resp., $f'\in E_v$ in $S_v$ where $(e,e'),(f,f')\in \R_{|S_v}$.
  Since $S_v\subseteq H$, 
	the edge $e$ is the opposite edge of $e'$ and  $f$ is the opposite edge of $f'$ in 
	a square which is also contained in $H$. Since $S_v$ is an isometric subgraph of $H$
  we can conclude that this square is chordless in $H$
  and thus  $(e,e'), (f,f')\in\delta(H)$. 
	Since $\R_{|S_v}$ is transitive it holds, $(e',f')\in \R_{|S_v}$.
	By analogous arguments as before we have $(e',f')\in \delta(H)$ and
  therefore, $(e,f)\in \delta^*(H)$.

 Finally, suppose $e\in E_v$ is a primal edge, $f\in F_v$ is
	non-primal and $(e,f)\in \R_{|S_v}$.
  By Lemma \ref{lem:PSP2}, $f$ is opposite to exactly one
	primal edge $e'$ where $(f,e')\in \R_{|S_v}$. If $e=e'$, then $e$ and $f$
	are opposite edges in a chordless square in $S_v$.
  By analogous arguments as before, we can conclude that
	this square is chordless in $H$ and hence, $e,f\in \delta(H)$.
   If $e\neq e'$, then $(e,f),(f,e')\in \R_{|S_v}$
	implies that $(e,e')\in \R_{|S_v}$
  and we can conclude from Corollary
	\ref{cor:PSP1} that there is no square spanned by $e$ and $e'$ in $S_v$.
	Again $e$ and $e'$ lie in common layer $\mathbb S_i^v$ and do not span
	any square in $H$. Thus we have $(e,e')\in\delta(H)$. Again,
	since $e'$ and $f$ are opposite edges in a chordless square in $H$
  we can conclude that $(e',f)\in\delta(H)$.
  Consequently, $\R_{|S_v}\subseteq\delta^*(H)$.
  Note, by results of Imrich \cite{Imrich:94} we have $\delta(H)^* \subseteq \sigma(H)$.
	It is easy to see that the connected components of $\delta(H)^*$ w.r.t. to a fixed equivalence
	class $i$ correspond to the layers of the factor $\mathbb S_i$. Therefore, we can conclude that
	$\delta(H)^* = \sigma(H)$.
	Hence, we have
   $$\R_{|S_v}=\R^*_{v|S_v}\subseteq\delta(H)^* = \sigma(H).$$

  Moreover, by Definition \ref{def:starfactor} of the star factors and since stars are prime,
	the number of $\R_{|S_v}$ classes equals the number of
	prime factors of $H$. Hence, it holds that $\R_{|S_v}$ and $\sigma(H)$ have the
	same number of equivalence classes.
\end{proof}

By the construction of star factors, the Distance Lemma and Theorem
\ref{thm:isomSubgraph}, we can directly infer the next corollary.

\begin{corollary}
	Let $G=(V,E)$ be an arbitrary given graph, $S_v$ be a PSP for some vertex $v\in V$
	and $\R_v^*$ have $k=1$ or $2$ equivalence classes.
	Then $$S_v \simeq \Box_{i=1}^k \mathbb S_i.$$
\end{corollary}

We conclude this section with a last theorem which shows that the
transitive closure of the union $\R_{|S_v}(V)$ over all vertices and its
relations $\R_v$, even restricted to $S_v$, is $\delta(G)^*$.

\begin{theorem}
	Let $G=(V,E)$ be a given graph and
	$\R_{|S_v}(V) = \cup_{v_\in V} \R_{|S_v}$.
	Then $$\R_{|S_v}(V)^* = \delta(G)^*.$$
\label{thm:union_equals_delta}
\end{theorem}
\begin{proof}
	By definition $\R_v \subseteq \delta(G)$.
	Moreover, by definition and Lemma
	\ref{lem:restrictionEquirel} it holds 
	that $\R_{|S_v}\subseteq \R_v^* \subseteq \delta(G)^*$ for all $v\in
	V(G)$. Thus, $\R_{|S_v}(V) \subseteq \delta(G)^*$, and hence
	$\R_{|S_v}(V)^*\subseteq \delta(G)^*$.		

	Let $e,f\in E(G)$ be edges that are in relation $\delta(G)$. By
	definition, $(e,f)\in \R_v$ for some $v\in V(G)$. If $e=(u,v)$ and
	$f=(w,v)$ are adjacent, then $e$ and $f$ are contained in the set $E_v$
	of $S_v$ and therefore in $\R_{|S_v} \subseteq \delta(G)^*$. Assume,
	$e=(u,v)$ and $f=(x,y)$ are opposite edges of a chordless square
	containing the edges $e,f$ and $g=(v,x)$. For contradiction, assume
	$(e,f)\not\in \R_{|S_v}(V)^*$ and hence $(e,f)\not\in \R_{|S_v}(V)$. Thus,
	for each $v\in V$ we have $(e,f)\not\in \R_{|S_v}$ and therefore, by
	definition, there is no square spanned by edges $e,e'\in E_v$ with
	$(e,e')\not\in \R_v^*$ such that $f$ is the opposite edge of $e$. In
	particular, this implies $(e,g)\in \R_v^*$ and hence $(e,g)\in \R_{|S_v}$.
	Analogously, one shows that $(f,g)\in \R_{|S_x}$. Since $\R_{|S_v} \cup
	\R_{|S_x} \subseteq \R_{|S_v}(V)$ we can infer that $(e,f) \in
	\R_{|S_v}(V)^*$, a contradiction.
\end{proof}

Theorem \ref{thm:union_equals_delta} allows us to provide covering algorithms for the
recognition of $\delta(G)^*$ or of $\delta(H)^*$ for subgraphs $H\subseteq
G$ that are based only on coverings by partial star products. Note, if
$\sigma(G) = \delta(G)^*$, then the covering of $G$ by partial star products
would also lead to a valid prime factorization. However, as most graphs are
prime we will in the next section provide algorithms, based on factorizable
parts, i.e., of coverings where the PSP's have more than one equivalence
class $\R_{|S_v}$, which can be used to recognize approximate products.

\section{Recognition of Relations, Colorings and Embeddings into Cartesian Products}
	
In order to compute local colorings based on partial star products and to
compute coordinates that respect this coloring
we begin with algorithms for the recognition of
$\R_{|S_v}(W)^*$ and $\delta(G)^*$.

\begin{lemma}
 Given a graph $G=(V,E)$ with maximum degree $\Delta$ and a
 subset $W\subseteq V$ such that $\la W \ra$ is
 connected, then Algorithm \ref{alg:LocalRW} computes
 $\R_{|S_v}(W)^*$ and $\cup_{v\in W} S_v$ in $O(|V|\Delta^4)$ time.
\label{lem:algoLocalRW}
\end{lemma}
\begin{proof}
 The Algorithm scans the vertices in an arbitrary order and computes
 $\langle N_2^G[v] \rangle$, $\delta'=\delta(\langle N_2^G[v] \rangle)$, as
 well as $S_v$ and $\R_{|S_v}$ w.r.t. $\delta'$. In order to compute the
 transitive closure of $\R_{|S_v}(W)$ an auxiliary graph, the color graph
 $\Gamma$, is introduced. For each vertex $v$ and to each equivalence class
 of $\R_{|S_v}$ some unique color is assigned, and $\Gamma$ keeps track of
 the ``colors'' of the equivalence classes. All vertices of $\Gamma$ are
 pairs $(e,c)$. Two vertices $(e',c')$ and $(e'',c'')$ are connected by an
 edge if and only if there is an edge $e \in \varphi_{c'} \cap
 \varphi_{c''}$ with $\varphi_{c'}\eqcl \R_{|S_u}$ and $\varphi_{c''}\eqcl
 \R_{|S_w}$ for some $u,w\in W$. In other words, if there is an edge $e$
 that obtained both, color $c'$ and $c''$. Edges in $\Gamma$ ``connect''
 edges of local equivalence classes that belong to the same global
 equivalence classes in $\R_{|S_v}(W)^*$. The connected components $Q$ of
 $\Gamma$ define edge sets $E_Q = \cup_{(e,c)\in Q} \varphi_c$. We
 therefore can identify the transitive closure of $\R_{|S_v}(W)^*$ by
 defining $e\in \varphi_Q \eqcl \R_{|S_v}(W)^*$ if $e\in E_Q$.
 Finally, we observe that this is iteratively done for all vertices $v\in
 W$, that all edges in $E(\la W\ra)$ are contained in some $E_v$ of $S_v$
 and, by Lemma \ref{lem:vMeetsEveryClass}, that every equivalence class of
 $\R_{|S_v}(W)^*$ is met by every vertex $v\in W$. Therefore, we can
 conclude that each edge is uniquely assigned to some class $\varphi_Q
 \eqcl \R_{|S_v}(W)^*$. Hence, the algorithm is correct.

 In order to determine the time complexity we first consider line
 \ref{alg:LocalRW:computeFirstStuff}. The induced $2$-neighborhood can be
 computed in $\Delta^2$ time and has at most $\Delta^2$ vertices, and hence
 at most $\Delta^4$ edges. As shown by Chiba and Nishizeki \cite{CN85} all
 triangles and all squares in a given graph $G=(V,E)$ can be computed in
 $O(|E|\Delta )$ time. Combining these results, we can conclude that all
 chordless squares can be listed in $O(|E|\Delta )$ time. Thus, in this
 preprocessing step, we are able to determine $\delta', S_v$ and $\R_{|S_v}$
 in $O(\Delta^4)$ time. Since this is done for all vertices $v\in W$, we end
 in an overall time complexity $O(|E|\Delta + |W|\Delta^4)$ for the
 preprocessing step and the while-loop. For the second part, we observe that
 $\Gamma$ has at most $O(|E|)$ connected components. Since the number
 of edges is bounded by $|V|\Delta$ we conclude that Algorithm
 \ref{alg:LocalRW} has time complexity $O(|V|\Delta^2 + |W|\Delta^4) =
 O(|V|\Delta^4)$.
\end{proof}

\begin{algorithm}[tbp]
\caption{\texttt{Local $\R_{|S_v}(W)^*$ computation}}
\label{alg:LocalRW}
\begin{algorithmic}[1]
\renewcommand{\baselinestretch}{1.1}
\vspace{1mm}
    \STATE \textbf{INPUT:} A graph $G=(V,E)$, $W \subseteq V$.
	\STATE $\sigma \gets W$
    \STATE initialize graph $\Gamma = \emptyset$; \COMMENT{called ``color graph''}
    \WHILE{$\sigma \neq \emptyset $}
    \STATE take any vertex $v$ of $\sigma$;
         \STATE compute $\langle N_2^G[v] \rangle$, $\delta'=\delta(\langle N_2^G[v] \rangle)$,
				 $S_v$ and $\R^*_{|S_v}$ w.r.t. $\delta'$; \label{alg:LocalRW:computeFirstStuff}
		 \STATE color the edges of $S_v$  w.r.t. the equivalence classes of $\R_{|S_v}$;
         \STATE set $num\_class$ = the number of equivalence classes of $\R_{|S_v}$;
         \STATE add $num\_class$ new vertices to $\Gamma$;
         \FOR {every edge $e$ in $S_v$}
            \IF{$e$ was already colored in $G$}
								\STATE x = old color of $e$; y = new color of $e$;
								\STATE add vertices $(x,e)$ and $(y,e)$ to $\Gamma$
                \STATE join all vertices of the from $(x,f)$ and $(y,f')$ in $\Gamma$;
             \ENDIF
         \ENDFOR
    \STATE delete $v$ from $\sigma$;
    \ENDWHILE
	\STATE	\COMMENT{compute the equivalence class $\varphi_k \eqcl \R_{|S_v}(W)^*$.}
    \STATE set $num\_comp$ = number of connected components of $\Gamma$;
    \FOR{$k=1$ to $num\_comp$}
			\IF{color of $e$ is vertex in component $k$ of $\Gamma$}
					\STATE $\varphi_k \gets e$;			
			\ENDIF
      \ENDFOR
 			\STATE \textbf{OUTPUT:} $\R_{|S_v}(W)^*$ and $\cup_{v\in W} S_v$;
\renewcommand{\baselinestretch}{1.}
\small\normalsize
\end{algorithmic}
\end{algorithm}

By means of Theorem \ref{thm:union_equals_delta} and Lemma
\ref{lem:algoLocalRW} we can directly infer the next corollary.

\begin{corollary}
 Let $G=(V,E)$ be a given graph with maximum degree $\Delta$.
 Then $\delta(G)^*$ can be computed in $O(|V|\Delta^4)$ time by a call
 of Algorithm \ref{alg:LocalRW} with input $G$ and $W=V$.
\label{cor:algoLocalRW}
\end{corollary}

As mentioned before, a vertex $x$ of a Cartesian product $\Box_{i=1}^n G_i$ is
properly ``coordinatized'' by the vector $c(x) := (c_1(x),\dots,c_n(x))$,
whose entries are the vertices $c_i(x)$ of its factor graphs $G_i$. Two
adjacent vertices in a Cartesian product graph differ in exactly one
coordinate. Furthermore, the coordinatization of a product is equivalent to
an edge coloring of $G$ in that edges $(x,y)$ share the same color $c_k$
if $x$ and $y$ differ in the coordinate $k$. This colors the edges of $G$
(with respect to the given product representation).

Conversely, the idea of Algorithm \ref{alg:findH} is to compute vertex
coordinates of a subgraph of $\cup_{v\in W} S_v$ based on its
$\R_{|S_v}(W)^*$-coloring. In particular, we want to compute coordinates
that reflect parts of the $\R_{|S_v}(W)^*$-coloring of $\cup_{v\in W} S_v$
in a consistent way. Consistent means that all adjacent vertices $u$ and
$v$ with $(u,v)\in \varphi_r \eqcl \R_{|S_v}(W)^*$ differ exactly in their
$r$-th position of their coordinate vectors, and no two distinct vertices obtain
the same coordinate. This goal cannot always be achieved for all vertices
contained in $\cup_{v\in W} S_v$. In \cite[p. 280 et seqq.]{Hammack:2011a}
a way is shown how to avoid those inconsistencies. In this approach colors
of edges with ``inconsistent'' vertices are merged to one color. However,
if the graph under investigation is only slightly perturbed, but prime, this
approach would merge all colors to one. This is what we want to avoid.
Instead of merging colors and hence, in order to preserve a possibly
underlying product structure, we remove those vertices in $\cup_{v\in W}
S_v$ where consistency fails. This leads to a subgraph $H\subseteq
\cup_{v\in W} S_v$ where the edges are still $\R_{|S_v}(W)^*$-colored w.r.t.
$\cup_{v\in W} S_v$ and have the desired coordinates. In Algorithm
\ref{alg:Embedding} we finally compute $H_i$ based on these coordinates and
the edges of $\varphi_i \eqcl (\R_{|S_v}(W)^*)_{|H}$, $1\leq i\leq k$.
Hence, the connected component of $H$ induced by the edges of $\varphi_i \eqcl
\R_{|S_v}(W)^*$ are subgraphs of layers $H_i$ of the Cartesian product
$\Box_{i=1}^{k} H_i$ and therefore, $H$ can be embedded into
$\Box_{i=1}^{k} H_i$.

\begin{algorithm}[htbp]
\caption{\texttt{Compute vertex coordinates of $H\subseteq\cup_{v\in W} S_v \subseteq G$}}
\label{alg:findH}
\begin{algorithmic}[1]
\renewcommand{\baselinestretch}{1.1}
\vspace{1mm}
    \STATE \textbf{INPUT:} A graph $G=(V,E)$, $W \subseteq V$; 	
    \STATE compute $\R_{|S_v}(W)^*$ and $\cup_{v\in W} S_v$ with \texttt{Local $\R_{|S_v}(W)^*$ computation} and input $G, W$; \label{alg:findH:initStart}
    \STATE $H\gets \cup_{v\in W} S_v$; \COMMENT{Note $W\subseteq V(H)$};
    \STATE $GoOn\gets true$ 	
 	\WHILE{$GoOn$} \label{alg:findH:whileStart}
	  \STATE $num\_class \gets$  number of equivalence classes of $\R_{|S_v}(W)^*$;
      \STATE $Q_i \gets$ subgraph of $H$ induced by edges of $\varphi_i \eqcl \R_{|S_v}(W)^*$ for all $i=1$ to $num\_class$;
      \STATE $Q_i(x) \gets$ connected component of $Q_i$ containing vertex $x$ for each $x\in V(H)$ for all $i=1$ to $num\_class$; 	
      \IF{exist $i$ and $j$ with $|V(Q_i(x))\cap V(Q_j(x))|>1$ for some $x\in V(H)$}
		  \STATE combine $\varphi_i$ and $\varphi_j$, i.e., compute $\varphi_i \cup \varphi_j$ in $\R_{|S_v}(W)^*$; \label{merge}
	  \ELSE \STATE $GoOn \gets false$;		
     \ENDIF	
	\ENDWHILE	\label{alg:findH:whileEnd}
    \STATE $v_0\gets$ arbitrary vertex of $W$;    \label{alg:findH:coordinatesSTART}
    \STATE label each vertex $x$ in each $Q_i(v_0)$ uniquely with $l_i(x)\in \{1, \dots ,|Q_i(v_0)|\}$; \label{alg:findH:initEnd}
    \STATE set coordinates $c_r(v_0)=0$ for all $r=1, \dots, num\_class$ \label{alg:findH:Coord1Start}
	\FOR{every vertex $x \in Q_i(v_0)$ and for all $i=1$ to $num\_class$}
		\STATE set coordinates $c_r(x)=0$ for all $r=1, \dots, num\_class$ and $r\neq i$;
		\STATE set coordinates $c_i(x)=l_i(x)$;
	\ENDFOR	\label{alg:findH:Coord1End}
  	\STATE $d_{\max}\gets \max_{x \in V(H)} d_H(v_0,x)$; \label{alg:findH:DistStart}
    \STATE $L_i \gets \{x \in V(H) \mid d_H(v_0,x)=i\}$ for $i=1,\dots d_{\max}$ \label{alg:findH:DistEnd}
	\FOR{$i=2$ to $L_{\max}$} \label{alg:findH:Scan}
	  \FOR{all $x\in L_i$ that have not obtained coordinates yet} \label{alg:findH:Scan2}
		\IF{for all $u\in N^H(x)$ that already obtained coordinates holds $(x,u)\in \varphi_r$ for some fixed $r$} \label{CoordCase1}
			\STATE	set coordinate $c_r(x) = l_r(x)$ \COMMENT{$l_r(x)$ is unique unused label};
			\STATE set coordinates $c_i(x)=c_i(u)$ for all $i=1, \dots, num\_class$, $i\neq r$;
    	\ELSIF{for all $u\in N^H(x)$ holds $u$ has not obtained coordinates} \label{CoordCase2}
			  \STATE remove $x$ and all edges adjacent to $x$ from $H$;
		      \STATE remove $x$ from $L_i$;	
		\ELSE   \label{CoordCase3}
         \STATE \COMMENT{now there are distinct neighbors $u,w\in N^H(x)$ and thus, have not been removed from $H$, such that
									they already obtained coordinates
        				  with $((x,u),(x,w)) \not\in \R_{|S_v}(W)^*$, i.e., $(x,u)\in \varphi_r$, $(x,w) \in \varphi_s$, $r\neq s$}
		   \STATE set coordinate $c_r(x)=c_r(w)$; set coordinate $c_s(x)=c_s(u)$;
     	   \STATE set coordinates $c_i(x)=c_i(u)$ for all $i=1$ to $ num\_class$, $i\neq r,s$;
		\ENDIF							
			\STATE call \texttt{ConsistencyCheck} for $x$ and vertices that already obtained coordinates;\label{consCheck}
	   \ENDFOR
	 \ENDFOR   \label{alg:findH:ScanEnd}
	\STATE \COMMENT{$H$ has been modified via deleting vertices $x$ that fail the consistency checks.}
    \STATE \textbf{OUTPUT:} $H$ with coordinatized vertices;
\renewcommand{\baselinestretch}{1.}
\small\normalsize
\end{algorithmic}
\end{algorithm}

\begin{lemma}
 Given a graph $G=(V,E)$ with maximum degree $\Delta$ and
  $W\subseteq V$ such that $\la W \ra$ is
 connected, then Algorithm \ref{alg:findH} computes the coordinates of a
 subgraph $H\subseteq G$ with $H\subseteq \cup_{v\in W} S_v$ such that
 \begin{enumerate}
	\item no two vertices of $H$ are assigned identical coordinates and
	\item adjacent vertices $x$ and $y$ with $(x,y)\in \varphi_r\eqcl \R_{|S_v}(W)^*$ differ
           exactly in the $r$-th coordinate.
 \end{enumerate}	
 The time complexity of Algorithm \ref{alg:findH} is $O(|V|\Delta^4 + |V|^2\Delta^2)$. 
\label{lem:alg:findH}
\end{lemma}
\begin{proof}
	The init steps (Line \ref{alg:findH:initStart} - \ref{alg:findH:initEnd})
	include the computation of $\R_{|S_v}(W)^*$, $H=\cup_{v\in W} S_v$, and
  the connected components $Q_i(x)$ that contain vertex $x$ and which are
	induced	by edges of $\varphi_i \eqcl \R_{|S_v}(W)^*$.
	By merging equivalence classes (Line \ref{merge}) we ensure that
	after the first while-loop connected components induced by
	$\R_{|S_v}(W)^*$ equivalence classes intersect in at most one vertex.
	Hence, vertices $x$ in $Q_i(v_0)$ can be assigned a unique label $l_i(x)$
	for each $i=1, \dots , num\_class$. In Line
	\ref{alg:findH:Coord1Start}-\ref{alg:findH:Coord1End} we assign coordinates
	to each vertex contained in $Q_i(v_0)$ for each $i=1, \dots, num\_class$.
	Since any two distinct subgraphs $Q_i(v_0)$ and $Q_j(v_0)$ intersect only
	in vertex $v_0$ we can ensure that adjacent vertices in each subgraph
	$Q_i(v_0)$ differ exactly in the $i$-th position of their coordinate. We
	finally compute the distances from $v_0$ to all other vertices in $H$, and
	distance levels $L_i$ containing all vertices $x$ with $d_H(v_0,x)=i$
	(Line \ref{alg:findH:DistStart} and \ref{alg:findH:DistEnd}). Notice,
	the preceding procedure assigns coordinates to all vertices of distance
	level $L_1$.

	In Line \ref{alg:findH:Scan} we scan all vertices in breadth-first search
	order w.r.t. to the root $v_0$, beginning with vertices in $L_2$, and
	assign coordinates to them.
 This is iteratively done for all vertices in level $L_i$ which either
 obtain coordinates based on the coordinates of adjacent vertices
 or are removed from graph $H$ and level $L_i$. In particular, in the subroutine
	\texttt{ConsistencyCheck} (Algorithm	\ref{alg:ConsistencyCheck}) we might
	also delete vertices and therefore we have to consider three cases.	\\
	\emph{First Case (Line \ref{CoordCase1}):} We assume that \emph{all}
	neighbors of a chosen vertex $x\in L_i$ that already obtained coordinates
	are contained in the \emph{same} subgraph $Q_r(x)$. Hence, the
	coordinates of $x$ should differ from their neighbor's coordinates in the
	$r$-th position. This is achieved by setting $c_r(x)$ to the unique label
	$l_r(x)$ and the rest of its coordinates identical to its neighbors.\\
	\emph{Second Case (Line \ref{CoordCase2})}:
  It might happen that vertex $x$ does not have any neighbor
	with assigned coordinates, that is, either those neighbors of $x$ are
  removed from $H$ and	$L_j$, $j\leq i$ in some previous step,
	or they have not obtained coordinates so far.
  If this case
	occurs, then we also remove vertex $x$ from $H$ and $L_i$, since no information
	to coordinatize vertex $x$ can be inferred from its neighbors.\\
	\emph{Third Case (Line \ref{CoordCase3}):}
	Let
	$u,w\in N^H(x)$ be neighbors of $x$ such that $u$ and $w$
	have already assigned coordinates and
	the edges $(x,u)$ and $(x,w)$ are in
	different equivalence classes. Assume $(x,u)\in \varphi_r$ and $(x,w) \in
	\varphi_s$, $r\neq s$.
	Keep in mind that $x$ should then differ from
 	$u$ and $w$ in the $r$-th and in the $s$-th position of its coordinates, respectively.
		Thus, we set coordinate $c_r(x)=c_r(w)$ and
	$c_s(x)=c_s(u)$. The remaining coordinates of $x$ are chosen to be
	identical to the coordinates of $u$. Note, we
	basically follow in this case the strategy to coordinatize vertices as
	proposed in \cite{AuHaIm92}.
     	
	In order to ensure that no two vertices obtained the same coordinates or
	that two adjacent vertices differ in exactly one coordinate we provide a
	consistency check in Line \ref{consCheck} and Algorithm
	\ref{alg:ConsistencyCheck}. If $x$ has the same coordinate as some
	previous coordinatized vertex we remove $x$ from $H$ and $L_i$. If $x$
	has a neighbor $y$ with coordinates that differ in more than one position
	from the coordinates of $x$ we delete the edge $(x,y)$ from $H$.

  To summarize, we end up with a subgraph $H\subseteq \cup_{v\in W} S_v$,
	such that the vertices of $H$ are uniquely coordinatized and such that
	adjacent edges $(x,y)\in\varphi_r\eqcl \R_{|S_v}(W)^*$ differ exactly in
	the $r$-th position of their coordinates.

	We complete the proof by determining the time complexity of Algorithm
	\ref{alg:findH}. Lemma \ref{lem:algoLocalRW} implies that Algorithm
	\ref{alg:LocalRW} determines $\R_{|S_v}(W)^*$ and $\cup_{v\in W} S_v$ in
	$O(|V|\Delta^4)$ time. Since $\la W\ra$ is connected, Lemma
	\ref{lem:vMeetsEveryClass} implies that $\R_{|S_v}(W)^*$ has at most
	$\Delta$ equivalence classes and therefore, the while-loop (Line
	\ref{alg:findH:whileStart} - \ref{alg:findH:whileEnd}) runs at most
	$\Delta$ times. The computation of the graphs $Q_i$ and $Q_i(x)$ within
	this while-loop can be done via a breadth-first search in $O(|E|+|V|) =
	O(|V|\Delta)$ time, since there are at most $|V|\Delta$ edges and
	connected components.
	The
	intersection and the union of $Q_i$ and $Q_j$ can be computed in
	$O(|V|^2)$. Hence, the overall-time complexity of the while-loop is
	$O(\Delta |V|^2)$.
	The assignments of coordinates to vertices $x\in
	Q_i(v_0)$ can be done in $O(\Delta)$ time. Since there are at most $|V|$
	vertices and at most $\Delta$ equivalence classes we end in
	$O(|V|\Delta^2)$ time. Computing distances from $v_0$ to all other
	vertices and the computation of $L_i$ can be achieved via breadth-first
	search in $O(|E|+|V|) = O(|V|\Delta)$ time. Consider now the two for-loops in
	Line \ref{alg:findH:Scan} and \ref{alg:findH:Scan2}.
	Each vertex is traversed exactly once. Hence these for-loops
	run $O(|V|)$ times.
	For each vertex in each distance levels we check whether
	there are neighbors in level $L_{i-1}$, which are at most $\Delta$ for
	each vertex $x$, and compute the $\Delta$ positions of the coordinates
	for each such vertex. The consistency check (Algorithm
	\ref{alg:ConsistencyCheck}) runs in $O(|V|(\Delta+\Delta)) =
	O(|V|\Delta)$ time. Hence, the overall time complexity of the for-loop
	(Line \ref{alg:findH:Scan} - Line \ref{alg:findH:ScanEnd})
	is $O(|V|^2\Delta^2)$.

	Combining these results, one can conclude that the time complexity
	of Algorithm \ref{alg:findH} is $O(|V|\Delta^4 + |V|^2\Delta^2)$. 
\end{proof}

\begin{algorithm}[tbp]
\caption{\texttt{ConsistencyCheck}}
\label{alg:ConsistencyCheck}
\begin{algorithmic}[1]
\renewcommand{\baselinestretch}{1.1}
\vspace{1mm}
    \STATE \textbf{REQUIRE:} Call \texttt{ConsistencyCheck} for vertex $x$ from Algorithm \ref{alg:findH}; 	
    \STATE \textbf{ENSURE:} no two vertices obtain identical coordinates and adjacent vertices differ in exactly one coordinate;
		    \FOR{all $y\in V(H)$, $x\neq y$ that already obtained coordinates} \label{consCheckStart}
			  \STATE \COMMENT{consistency check that no two vertices obtain the same coordinates}
              \IF{$c_r(x)=c_r(y)$ for all $r=1$ to $num\_class$}
                   \STATE remove $x$ and all edges adjacent to $x$ from $H$;
				   \STATE remove $x$ from $L_i$;
				   \STATE break for loop;
              \ELSE
				\STATE \COMMENT{consistency check that two adjacent vertices differ only in one $r$-th coordinate}
				\IF{$(x,y)$ is edge contained in some $\varphi_r$ and 										
				    	$c_r(x)=c_r(y)$ or $c_i(x)\neq c_i(y)$ for some $i=1$ to $num\_class$, $i\neq r$ } \label{alg:ConsistencyCheck:deleteEdge}
                        \STATE remove edge $(x,y)$ from $H$;
	    				\STATE break for loop;
				   \ENDIF
			    \ENDIF
			 \ENDFOR
\renewcommand{\baselinestretch}{1.}
\small\normalsize
\end{algorithmic}
\end{algorithm}

\begin{algorithm}[tbp]
\caption{\texttt{Embedding of $H$ into Cartesian product}}
\label{alg:Embedding}
\begin{algorithmic}[1]
\renewcommand{\baselinestretch}{1.1}
\vspace{1mm}
    \STATE \textbf{INPUT:} A graph $G=(V,E)$ with coordinatized vertices; 	
    \FOR{each position $i=1$ to $r$ of coordinates}
		\STATE initialize graph $H_i=\emptyset$;
		\FOR{each vertex $v\in V$}
			\IF{$c_i(v)\notin V(H_i)$}
				\STATE add $c_i(v)$ to $V(H_i)$;
			\ENDIF
		\ENDFOR
	\ENDFOR
	\FOR{each position $i=1$ to $r$ of coordinates}
	  \FOR{each edge $(x,y)\in E$}
			\IF{$c_i(x)\neq c_i(y)$ and edge $(c_i(x),c_i(y))\notin E(H_i)$}
				\STATE add $(c_i(x),c_i(y))$ to $E(H_i)$;
			\ENDIF
	  \ENDFOR
	\ENDFOR		
    \STATE \textbf{OUTPUT:} Factors $H_i$ and Cartesian product $\Box_{i=1}^r H_i$ where $G$ can be embedded into;
\renewcommand{\baselinestretch}{1.}
\small\normalsize
\end{algorithmic}
\end{algorithm}

\begin{figure}[tbp]
  \centering
  \subfigure[A Cartesian prime graph $G=(V,E)$ is shown. For all vertices
             $x\in V$ (marked with "X") the respective $\R_{|S_x}$ has only
             one equivalence class. Thus, we use only all non-"X"-marked
             vertices, pooled in the set $W\subseteq V$ and call
             \texttt{Local $\R_{|S_v}(W)^*$ computation} (Alg. \ref{alg:LocalRW}). The equivalence
             classes  of $\R_{|S_v}$ for vertex $v=v_0$ are highlighted by dashed and thick
             edges. ]{
    \label{fig:Labelname1}
    \includegraphics[bb=45 367 360 640, width=0.45\textwidth]{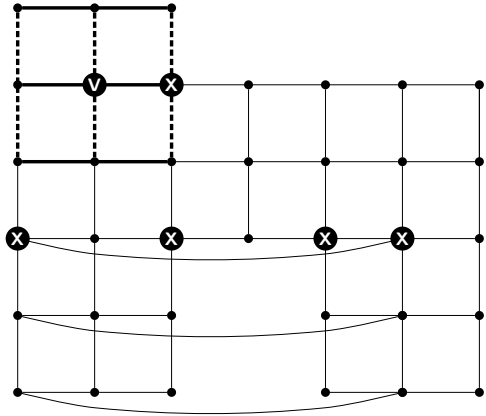}
  } $\qquad$
  \subfigure[After calling \texttt{Local $\R_{|S_v}(W)^*$ computation} (Alg. \ref{alg:LocalRW}) we
             obtain the equivalence classes of $\R_{|S_v}(W)^*$ highlighted
             by dashed and thick edges. After calling \texttt{Compute
             vertex coordinates} (Alg. \ref{alg:findH}, Line \ref{alg:findH:coordinatesSTART} -
             \ref{alg:findH:Coord1End}) we obtain a graph where the
             vertices in each $G_i^v$-layer obtain unique coordinates. ]{
    \label{fig:Labelname2}
\includegraphics[bb=45 367 360 640, width=0.45\textwidth]{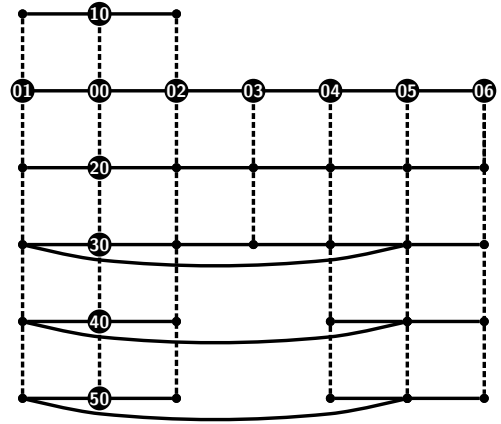}
  }
  \centering
  \subfigure[Shown is the graph $G$ with coordinatized vertices for all
             $x\in \cup_{i=1}^4 L_i$. Note, the vertex $x$ with coordinates
             $(37)$ obtained a new unused second coordinate $7$, since all
             edges $(u,x)$ where $u$ already obtained coordinates are from the same
             equivalence class (Alg. \ref{alg:findH}, Line \ref{CoordCase1}).
							Thus, coordinates cannot be combined.]{
    \label{fig:Labelname3}
    \includegraphics[bb=45 367 360 660, width=0.45\textwidth]{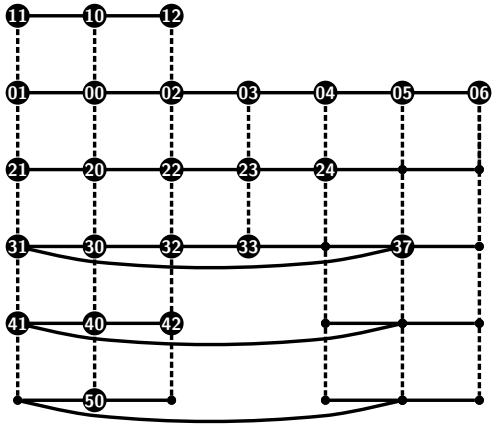}
  } $\qquad$
  \subfigure[Shown is the graph $G$ with coordinatized vertices for all
             $x\in \cup_{i=1}^5 L_i$. Note, after running
             \texttt{ConsistencyCheck} (Alg. \ref{alg:ConsistencyCheck},
						Line \ref{alg:ConsistencyCheck:deleteEdge})
             the edge between the vertices with
             coordinates $(37)$ and $(25)$ is deleted, since the vertices
             differ in more than one coordinate.]{
    \label{fig:Labelname4}
\includegraphics[bb=45 367 360 660, width=0.45\textwidth]{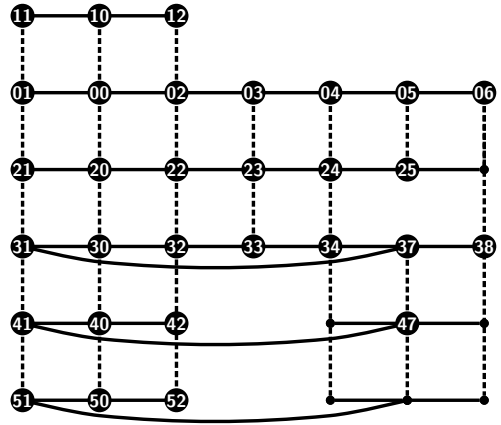}
  }
  \caption{The basic steps of Algorithm \ref{alg:LocalRW} and \ref{alg:findH}}
  \label{fig:ExmplAlgo}
\end{figure}

\begin{figure}[tbp]
	\centering
  \includegraphics[bb= 43 289 440 644, scale=0.5]{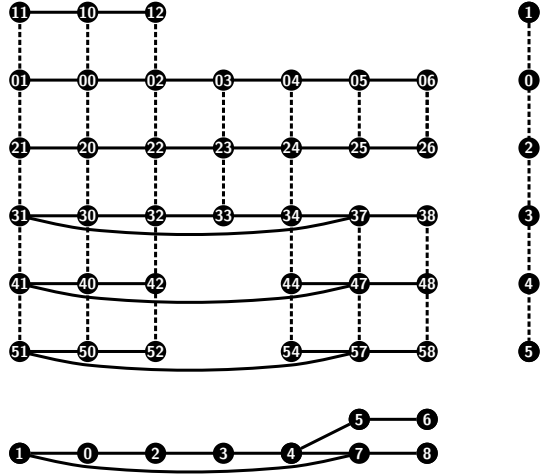}
 \caption{After running Algorithm \ref{alg:LocalRW} and \ref{alg:findH} we
           obtain $H$ as a subgraph of the graph $G$ in Figure
           \ref{fig:ExmplAlgo}, with coordinatized vertices, and edges colored w.r.t. $\R_{|S_v}(W)^*$
           equivalence classes. After running
           \texttt{Embedding of $H$ into Cartesian product} (Alg.
           \ref{alg:Embedding}) we obtain the putative factors $H_1$ and
           $H_2$ of $H$ and, hence, of $G$. Note, due to the
           coordinatization of $H$ the embedding of $H$ into $H_1 \Box H_2$
           can easily be determined. }
  \label{fig:ExmplAlgo2}
\end{figure}

\begin{lemma}
 	Given a graph $G=(V,E)$ with maximum degree $\Delta$ obtained from
	Algorithm \ref{alg:findH} with coordinatized vertices. Then Algorithm
	\ref{alg:Embedding} computes factors $H_i$ such that $G$ can be embedded
	into $\Box_{i=1}^r H_i$ in $O(|E|\Delta)$ time.
\label{lem:alg:Embedding}
\end{lemma}
\begin{proof}
	After running Algorithm \ref{alg:findH} we obtain a graph $G=(V,E)$ such
	that vertices $x\in V$ have consistent coordinates $c(x) = (c_1(x),
	\dots, c_r(x))$, i.e, no two vertices of $G$ have identical coordinates
	and adjacent vertices $x$ and $y$ with $(x,y)\in\varphi_i$ differ only in
	the $i$-th position of their coordinates.

  We first compute empty graphs $H_1,\dots,H_r$ and add for each vertex
  $x$ and for each $c_i(x)$ of its coordinates $c(x) = (c_1(x),\dots
  c_r(x))$ the vertex $c_i(x)$ to $H_i$. Different vertices $c_i(x)$ and
  $c_i(y)$ are connected in $H_i$ whenever there is an edge $(x,y)\in E$.
  We define a map $\gamma:V(G)\rightarrow V(H)$ with $x \mapsto c(x)$.
  Since no two vertices of $G$ have identical coordinates $\gamma$ is
  injective. Furthermore, since adjacent vertices $x$ and $y$ that differ
  only in one, say the $i$-th, position of their coordinates are mapped to
  the edge $(c_i(x), c_i(y))$ contained in factor $H_i$ and by definition
  of the Cartesian product, we can conclude that the map $\gamma$ is a
  homomorphism and hence, an embedding of $G$ into $H$.

    The first two for-loops run $|V|\Delta$ times, that is $O(|E|)$. The
  second two for-loops run $|E|\Delta$ times, hence we end in overall time
  complexity of $O(|E|\Delta)$.
\end{proof}

To complete the paper, we explain how the last algorithms, in particular,
Algorithm \ref{alg:LocalRW}, \ref{alg:findH} and \ref{alg:Embedding} can be
used as suitable heuristics to find approximate products; see also Figures
\ref{fig:ExmplAlgo} and \ref{fig:ExmplAlgo2}. Note, by Corollary
\ref{cor:algoLocalRW} Algorithm \ref{alg:LocalRW} can be used to compute
$\delta(G)^*$. However, most graphs are prime and $\delta(G)^*$ would
consist only of one equivalence class. Thus we are interested in subsets of
$\delta(G)^*$ which provide enough information of large factorizable or
``into non-trivial Cartesian product embeddable'' subgraphs. This can be
achieved by ignoring regions $S_v$ where $\R_{|S_v}$ has only one or less
than a given threshold number of equivalence classes. Hence, only subsets
$W\subseteq V$ where $\R_{|S_v}(W)^*$ has a sufficiently large number of
equivalence classes are of interest. For this, we would cover a graph by
starting at some vertex $v\in V$, compute $S_v$ and $\R_{|S_v}$, and check if
$\R_{|S_v}$ has the desired number of equivalence classes; see Figure
\ref{fig:Labelname1}. If not, we take another vertex $w\in V$ and repeat
this procedure with $w$. If $\R_{|S_v}$ has the desired number of
equivalence classes we would take a neighbor $w$ of $v$, compute $S_w$ and
$\R_{|S_w}$ and check whether $(\R_{|S_w} \cup \R_{|S_v})^*$ has the desired number
of equivalence classes. If so, then we continue with neighbors of $v$ and $w$ and
to extend the regions that can be embedded into a Cartesian product.
To find such regions one can easily adapt Algorithms \ref{alg:findH} and
\ref{alg:Embedding}.

Note, after running Algorithm \ref{alg:LocalRW} one could take out one of
largest connected component of each equivalence class induced by edges with
the respective ``colors'' to obtain putative factors; see Figure
\ref{fig:Labelname2}. However, even knowing putative factors does not
yield information about which edges
should be added or deleted to obtain a
product graph. For this, coordinates are necessary. They can be computed by
Algorithm \ref{alg:findH} and used as input for Algorithm \ref{alg:Embedding};
see Figure \ref{fig:ExmplAlgo2}.

Finally, even the most general methods for computing approximate strong products only compute
a (partial) product coloring of the graphs $G$ under investigation. They yield putative factors,
but no coordinatization \cite{hellmuth2011local}. However, Algorithm \ref{alg:Embedding}
can be adapted to find the coordinates of the so-called
underlying approximate Cartesian skeleton of such graphs, and can thus be used
to find an embedding of (the approximate strong product) $G$
into a non-trivial strong product graph.

\bibliographystyle{plain}
\bibliography{lib}

\end{document}